\documentclass[journal, 10pt, compsoc]{IEEEtran}

\usepackage[english]{babel}

\usepackage{algorithm}
\usepackage{algpseudocode}
\usepackage{algorithmicx}
\usepackage{amsmath}
\usepackage[numbers]{natbib}

\usepackage{amsthm}
\usepackage{amssymb}
\usepackage{mathtools}
\usepackage{bbm}
\usepackage{subfigure}

\usepackage{ulem}

\newtheorem{thm}{Theorem}
\newtheorem{defi}{Definition}

\newtheorem{lem}{Lemma}

\newtheorem{rmk}{Remark}
\newtheorem{exmpl}{Example}

\usepackage{tikz}
\usepackage{amsmath}

\title{Correction to ``Local Information Privacy and Its Applications to Data Aggregation"}
\author{Bo Jiang,~\IEEEmembership{Student Member,~IEEE,}
        Ming Li,~\IEEEmembership{Fellow, IEEE,}
        and~Ravi~Tandon,~\IEEEmembership{Senior Member,IEEE}
\IEEEcompsocitemizethanks{\IEEEcompsocthanksitem Bo Jiang, Ming Li and Ravi Tandon are with the Department
of Electrical and Computer Engineering, University of Arizona, Tucson,
AZ, 85721.\protect\\
E-mail: bjiang@email.arizona.edu, lim@email.arizona.edu, tandonr@email.arizona.edu}}

\begin{document}
\maketitle

\begin{abstract}
In our previous works \cite{LIP1,LIP2, jiang2018context}, we defined $(\epsilon,\delta)$-Local Information Privacy (LIP) as a context-aware privacy notion and presented the corresponding privacy-preserving mechanism. Then we claim that the mechanism satisfies $(\epsilon,0)$-LIP for any $\epsilon>0$ for arbitrary $P_X$. However, this claim is not completely correct. In this document, we provide a correction to the valid range of privacy parameters of our previously proposed LIP mechanism. Further, we propose efficient algorithms to expand the range of valid privacy parameters. Finally, we discuss the impact of updated results on our original paper's experiments, the rationale of the proposed correction and corrected results.

\end{abstract}

\section{valid range of privacy parameters}

\begin{defi} [$(\epsilon,\delta)$-Local Information Privacy]\label{def:LIP1}
A mechanism $\mathcal{M}$  satisfies $(\epsilon,\delta)$-LIP for some $\epsilon\in{\mathbf{R}^+}$ and $\delta\in[0,1]$, if $\forall{S_x\in{\mathcal{X}}}$, $S_y\in{\textit{Range}(\mathcal{M})}$:

\begin{equation}\label{cons0}
\begin{aligned}
    & \operatorname{Pr}(Y\in\mathcal{S}_y) \geq e^{-\epsilon}\operatorname{Pr}(Y\in{\mathcal{S}_y}|X\in\mathcal{S}_x)-\delta, \\
    &\operatorname{Pr}(Y\in\mathcal{S}_y)\le{e^{\epsilon}\operatorname{Pr}(Y\in{\mathcal{S}_y}|X\in\mathcal{S}_x)}+\delta.
\end{aligned}
\end{equation}
\end{defi}
Then, in Theorem 2 of \cite{jiang2018context}, we proposed the binary context-aware randomize response (RR) mechanism. Later in Theorem 3 of \cite{LIP2}, we extended this mechanism to the $M$-ary case. Then, in \cite{LIP1}, we mentioned this mechanism in Section IV-B. Specifically, the mechanism can be described as follows: let $P_X$ denote the prior distribution of the input data $X$. Then the context-aware randomize response mechanism takes $X$ as input and releases $Y$ as follows:

\begin{equation}\label{rr-basic}
\begin{aligned}
     q_{xy} = \begin{cases}  \frac{P_X(y)}{e^{\epsilon}}, & x\neq y, x, y\in \mathcal{X},   \\ 1- \frac{(1-P_X(x))}{e^{\epsilon}}
   , & x= y, x, y\in \mathcal{X}. 
   \end{cases} 
\end{aligned}
\end{equation}
We use $q_{xy}$ to denote the probability of releasing $y$ as the output when the input is $x$.
Then in our previous work of \cite{LIP1,LIP2, jiang2018context}, we claim that the mechanism defined in \eqref{rr-basic} satisfies $(\epsilon,0)$-LIP or pure LIP for any $\epsilon\ge0$ and for arbitrary $P_X$.

Specifically, in the next Theorem, we show that the mechanism satisfies $(\epsilon,0)$-LIP as long as $\epsilon$ is larger than a threshold.

\begin{thm}
   For a given probability distribution of $P_X$, with $P_{\min} = \min_x P_{X}(x)$, the mechanism defined in \eqref{rr-basic} satisfies $(\epsilon,0)$- LIP, as long as 
   \begin{equation*}
       \epsilon \ge \log \left(\frac{1}{P_{\min}} -1\right) ~\text{or}~~\epsilon = 0.
   \end{equation*}
\end{thm}




\begin{rmk}

We note that for any distribution $P_X$, $P_{\min}$ satisfies the inequality $P_{\min}\le 1/|\mathcal{X}|$. This implies that other than the special case of $\epsilon=0$, the smallest value of $\epsilon$ supported by our mechanism is $\log(|X|-1)$, which is only possible when $P_X$ is uniformly distributed. In general, when $P_X$ is not uniform, the smallest value of $\epsilon$ is a function of $P_{\min}$, which is larger than $\log(|\mathcal{X}| -1)$.
\end{rmk}

\begin{figure}[t]
\centering 
\subfigure[$\epsilon$-$\delta$ tradeoff for different $P_{\min}$.]
{\includegraphics[width=0.48\textwidth]{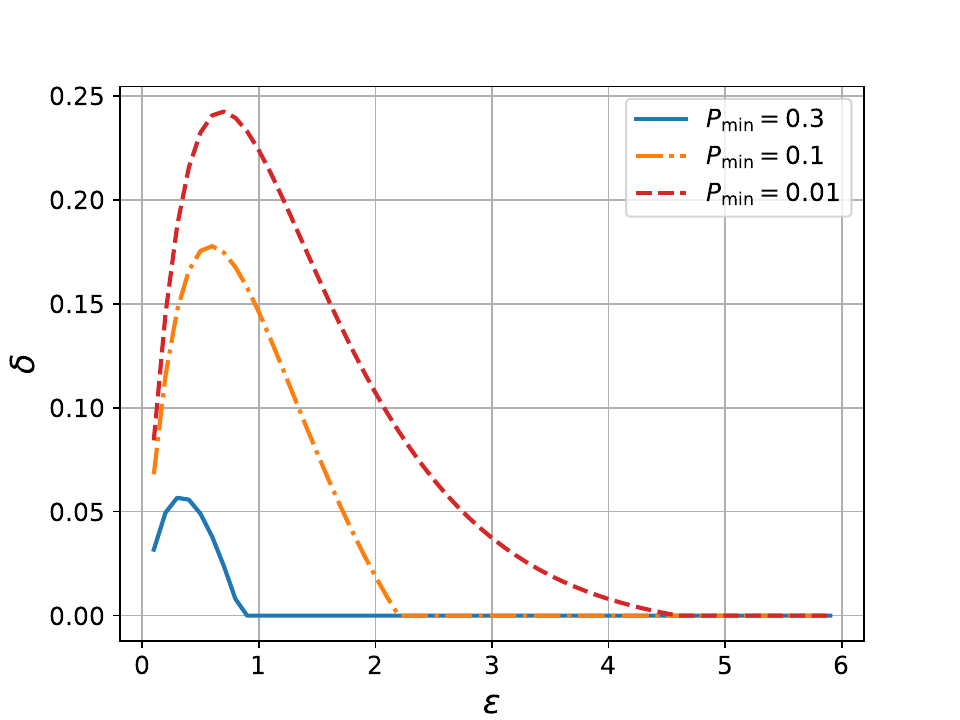}
\label{ep.vs_delta}}
\subfigure[$P_{\min}$-$\epsilon$ tradeoff for different $\delta$.]
{\includegraphics[width=0.48\textwidth]{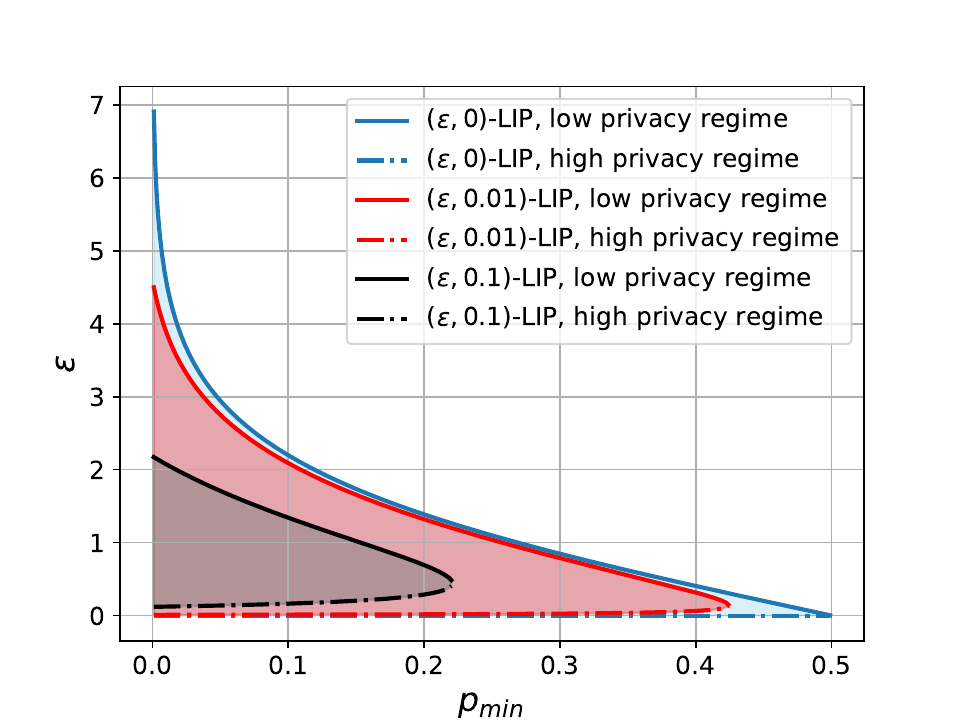}
\label{pmin_vs_ep.}}
\caption{When the context-aware randomize response mechanism achieves $(\epsilon,\delta)$-LIP, The relationship between $P_{\min}$ in the prior distribution and $\epsilon,\delta$. Shaded area represents the infeasible region of $\epsilon$ under $\delta$ and $P_{\min}$.}
\label{fig:var_ours_cms}
\end{figure}


We note that if the desired $\epsilon$ is smaller than $\log(1/P_{\min}-1)$, then  by allowing some failure probability in the definition to achieve a relaxed  $(\epsilon,\delta)$-LIP, the desired value of $\epsilon$ can be achieved. The $\epsilon$, $\delta$ tradeoff with a given distribution is given in the next Theorem.
\begin{thm}
   For a given probability distribution of $X$, with $P_{\min} = \min_x P_{X}(x)$, the context-aware randomized response mechanism achieves $(\epsilon, \delta)$- LIP, where 
   \begin{equation}
    \delta \ge \frac{(e^{\epsilon} -1)[1 - P_{\min}(e^{\epsilon} + 1)]}{e^{2\epsilon}}.
\end{equation}
\end{thm}
\begin{rmk}
We note that by introducing the failure probability $\delta$, the valid range of $\epsilon$ enlarges to:

Case 1: $\delta < \frac{1}{4(1-P_{\min})}-P_{\min} $:
\begin{small}
\begin{equation*}
\begin{aligned}
   \epsilon\in &\left[0,\log\left(\frac{1 - \sqrt{1-4(\delta+P_{\min})(1-P_{\min})}}{2(\delta +P_{\min})}\right)\right]\\
   \cup &\left.\left[\log\left(\frac{1 + \sqrt{1-4(\delta+P_{\min})(1-P_{\min})}}{2(\delta +P_{\min})}\right), \infty\right.\right).
\end{aligned}
\end{equation*}
\end{small}

Case 2: $\delta \ge \frac{1}{4(1-P_{\min})}-P_{\min}$
\begin{equation*}
    \epsilon \in [0, \infty].
\end{equation*}
\end{rmk}
In Figure 1, we plot $\epsilon-\delta$ tradeoff for different $P_{\min}$ and the upper bound and lower bound of valid epsilon according to the result in Theorem 2 and Remark 2 respectively. 

Specifically, in Figure \ref{ep.vs_delta}, we let $P_{\min}$ to be $0.3$, $0.1$ and $0.01$, respectively, and showcase the range of $\delta$ when picking an $\epsilon$ that varies from  $0$ to $6$. In Figure \ref{pmin_vs_ep.}, we  highlight the infeasible region of $\epsilon$ in the plot. We vary the value of $P_{\min} \in[0.001, 0.5]$ for $\delta = 0, 0.1$ and $0.01$ respectively.  Observe that enlarging $\delta$ can effectively reduce the infeasible region of $\epsilon$. 

Observe that the valid region of $\epsilon$ contains two parts, and in the following, we focus on the more practical low privacy regime where 
$$\epsilon\in \left.\left[\log\left(\frac{1 + \sqrt{1-4(\delta+P_{\min})(1-P_{\min})}}{2(\delta +P_{\min})}\right), \infty\right.\right).$$
We next present an effective algorithm to achieve $(\epsilon,\delta)$-LIP for a given $P_{X}$. The intuition comes from the following remark that highlights the relationship between the minimum $P_{\min}$ and a pair of $(\epsilon,\delta)$.

\begin{rmk}
By Theorem 1 and Theorem 2, to satisfy $(\epsilon,\delta)$-LIP for $\epsilon$ in the low privacy regime, the lower bound of $P_{\min}$ becomes 
\begin{equation}\label{eq:p_min}
    P_{\min} \ge \frac{1}{e^{\epsilon} +1} - \frac{\delta e^{2\epsilon}}{e^{\epsilon} - 1}.
\end{equation}
\end{rmk}

Note that to achieve $(\epsilon,\delta)$-LIP, a feasible approach is to artificially increase $P_{\min}$ until $P_{\min}$ is over the threshold as a function of $\epsilon$ and $\delta$. To this end, we propose the following grouping based algorithm, which combines inputs with small priors to lift $P_{\min}$, in order to achieve $(\epsilon,\delta)$-LIP for small $\epsilon$s.  We next describe the key steps in the algorithm. We take out original input PMF which has the alphabet $X$; and then create a new PMF $P'$ over a reduced alphabet $X'$; To this end, we consider a tunable parameter $\ell$, and combine $\ell$ inputs with smallest probabilities and denote this as a new ``grouped" input denoted as $x'$. The algorithm is described in Algorithm 1. 

 \begin{algorithm}[t]
\caption{Grouping based context-aware RR achieving $(\epsilon,\delta)$-LIP}
\hspace*{\algorithmicindent}
\textbf{Input:} $\epsilon$, $\delta$, $P_X$,  $\ell\in[1, |\mathcal{X}|]$.\\
 \hspace*{\algorithmicindent} \textbf{Output:} $Y$
\begin{algorithmic}[1]
    \State Sort the values in $\mathcal{X}$ according to their probabilities in an non-decreasing manner and obtain the set of instances:
    $$\{x_1, x_2,...,x_{|\mathcal{X}|}\};$$
    \State Group elements $\{x_1,...x_{\ell}\}$ into a new instance $x'$, with $$P_X(x') =\sum_{i=1}^{\ell} P_X(x_i);$$
    \State Update the support of $X$: $\mathcal{X}'\gets \{x', x_{\ell+1},...,x_{|\mathcal{X}|}\};$
    \State Release $x$ as follows:
    \If {$x \in \{x_1,...,x_{\ell}\}$}
    \State Release 
    \begin{equation*}
        Y = \begin{dcases}
          x', &~\text{w.p.}~ 1- \frac{(1-P_X(x'))}{e^{\epsilon}};\\
          y\neq {x'}, & ~\text{w.p.}~ \frac{P_X(y)}{e^{\epsilon}}.
        \end{dcases}
    \end{equation*}
    \Else
    \State Release 
    \begin{equation*}
        Y = \begin{dcases}
          x, &~\text{w.p.}~ 1- \frac{(1-P_X(x))}{e^{\epsilon}};\\
          y\neq {x}, & ~\text{w.p.}~ \frac{P_X(y)}{e^{\epsilon}}.
        \end{dcases}
    \end{equation*}
    \EndIf

    \State \Return $Y$
\end{algorithmic}
\label{algo:2}
\end{algorithm}



\begin{thm}
    The algorithm that takes $X\in\mathcal{X}$ as input and releases $Y\in\mathcal{X}'$ achieves $(\epsilon,\delta)$-LIP, where 
    \begin{equation*}
    \delta \ge \frac{(e^{\epsilon} -1)[1 - P'_{\min}(e^{\epsilon} + 1)]}{e^{2\epsilon}},
\end{equation*}
and $P'_{\min} = \min\{P_X(x'), P_X(x_{l + 1})\}$.
\end{thm}

\begin{rmk}
    Note that Algorithm 1 groups realizations of $X$ with small priors. The number of elements in the group depends on a tunable parameter $\ell$: a larger $\ell$ significantly increases $P_{\min}$, thereby achieving a smaller $\epsilon$ as the privacy parameter. However, a large $l$ could lead to collisions between a large number of inputs, making them indistinguishable in the output, which may result in reduced utility. Therefore, the parameter $\ell$ can be optimized to maximize utility for a given application
\end{rmk}
Next, we use the following example to illustrate the implementation of Algorithm 1.

\begin{exmpl}
Consider the input data $X$ has a support of $\{1, 2, 3, 4, 5\}$, and the prior distribution of $X$ is $P_X(1) = P_X(2) = 0.05$,  $P_X(3) = 0.2$, $P_X(4) = 0.3$, $P_X(5) = 0.4$. We next look into four cases on different mechanisms corresponding to pure LIP or $(\epsilon,\delta)$-LIP. For each case, we derive the minimal value of $\epsilon$.

\textbf{Case 1, Pure LIP with context-aware RR:}

With the original distribution, to achieve $(\epsilon,0)$-LIP, the minimal value of $\epsilon$ one can pick is:
$$\epsilon\ge \log(1/0.05 - 1) \approx 2.94,$$ corresponding to a $P_{\min} = 0.05$. 

\textbf{Case 2, Pure LIP with the group algorithm:}

When $\ell =2$, $X= 1$ and $X =2$ are grouped together, and their priors are combined to lift $P_{\min}$ to $0.1$. Now, to achieve $(\epsilon,0)$-LIP, the minimal $\epsilon$ becomes:
$$\epsilon\ge \log(1/0.1 - 1) \approx 2.19.$$

When  $\ell = 3$, $X=1$, $X=2$ and $X=3$ are grouped together, making $P_{\min} = 0.3$. Then, to achieve $(\epsilon,0)$-LIP, the minimal $\epsilon$ becomes:
$$\epsilon\ge \log(1/0.3 - 1) \approx 0.84,$$

\textbf{Case 3, Approximate LIP with context-aware RR:}

With the original distribution, to achieve $(\epsilon,0.01)$-LIP, the minimal value of $\epsilon$ one can pick is:
$$\epsilon\ge \log\left(\frac{1+\sqrt{1-4(0.01 + 0.05) (1-0.05)}}{2(0.01 + 0.05)}\right) \approx 2.75,$$ corresponding to a $P_{\min} = 0.05$.

\textbf{Case 4, Approximate LIP with the group Algorithm:}

When $\ell = 2$, as shown in Case 2, $P_{\min}$ is increased to $0.1$, then to
 achieve $(\epsilon,0.01)$-LIP, the minimal value of $\epsilon$ becomes
$$\epsilon\ge \log\left(\frac{1+\sqrt{1-4(0.01 + 0.1) (1-0.1)}}{2(0.01 + 0.1)}\right) \approx 2.08,$$

When $\ell = 3$, $P_{\min}$ becomes $0.3$. Therefore,
 to achieve $(\epsilon,0.01)$-LIP, the minimal value of $\epsilon$ becomes 
$$\epsilon\ge \log\left(\frac{1+\sqrt{1-4(0.01 + 0.3) (1-0.3)}}{2(0.01 + 0.3)}\right) \approx 0.78.$$
\end{exmpl}


\section{Discussion}

\begin{enumerate}

\item \textit{{Impact of updated results on our original paper's experiments: }} {The analysis presented in this document focuses solely on the valid range of $\epsilon$ and does not alter any of the mechanism's parameters. Consequently, the experimental results from the original paper remain valid, with the only change being an updated $\epsilon$ range recalculated based on the analysis in this correction.}
    \item \textit{The choice of $\epsilon,\delta$ for local mechanisms in the wild: }The real-world implementation of local privatization models usually adopts large $\epsilon$s to provide good utility. For example, Apple picks $\epsilon = 8$ in their LDP-based Count Mean Sketch protocol for domain collection, and $\epsilon=4$ for emoji collection \cite{281394}. This does not significantly degrade the privacy guarantee, as they use an Encryption - Shuffling - Analysis (ESA) architecture. In this architecture, the shuffler, positioned in the middle, prevents the server from tracking the original data source and provides privacy amplification by a factor of $1/\sqrt{N}$, where $N$ denotes the number of users. We have theoretically shown that $2\epsilon$-LIP implies $\epsilon$-LDP in proposition 1 of \cite{LIP1}, which means LIP can also benefit from shuffling. Even though the original work on shuffled DP does not include an analysis for the relationship between $(\epsilon,\delta)$-LDP and CDP, the idea of amplifying the privacy guarantee through shuffling remains the same. Additionally, the tightness of this benefit for LIP remains an interesting area for future research.

    \item \textit{The rationale of collision in Algorithm 1:} A natural limitation of the mechanism of general randomized response (GRR) \cite{203872} is the suboptimal utility-privacy tradeoff for high-dimensional data, corresponding to a small $P_{\min}$ in context-aware models. In LDP research, GRR is usually combined with hashing techniques that reduce input data dimension at the cost of slight data input collision. This technique has been proven to achieve a better utility-privacy tradeoff than GRR for high-dimensional data. A natural extension of this idea to LIP is to design functions that combine inputs with small priors until $P_{\min}$ is greater than $1/(e^{\epsilon} + 1)$ for a given $\epsilon$. This operation guarantees privacy and utility concurrently.
    \end{enumerate}














\appendices

\section{Proof of Theorem 1.}

In this section, we proof Theorem 1, i.e., the context-aware RR mechanism achieves $(\epsilon,0)$- LIP, where 
   \begin{equation*}
       \epsilon \ge \log \left(\frac{1}{P_{\min}} -1\right).
   \end{equation*}

We first introduce the following lemma which is helpful in the proof.
\begin{lem}\label{lem1}
    With the context-aware RR, the marginal distribution of $Y$ is equivalent to the marginal distribution of $X$, i.e. $P_X(x) = P_Y(x)$, for all $x \in\mathcal{X}$.
\end{lem}

\begin{proof}
\begin{equation*}
\begin{aligned}
    &P_Y(y)\\
    = &P_{Y|X}(y|y) P_{X}(y) +\sum_{x\neq y} P_{Y|X} (y|x) P_X(x)\\
    =& \frac{e^{\epsilon}-1+P_X(y)}{e^{\epsilon}} P_X(y) + \sum_{x\neq y} \frac{P_X(y)P_X(x)}{e^{\epsilon}}\\
    =&\frac{e^{\epsilon}-1+P_X(y)}{e^{\epsilon}} P_X(y) +  \frac{P_X(y)(1-P_X(y))}{e^{\epsilon}}\\
    =&\frac{e^{\epsilon}P_X(y)-P_X(y)+P_X(y)P_X(y)}{e^{\epsilon}}  +  \frac{P_X(y)-P_X(y)P_X(y)}{e^{\epsilon}}\\
    =&P_X(y).
\end{aligned}
\end{equation*}
\end{proof}

For $(\epsilon,0)$-LIP, Definition 1 can be further simplified to, for all $y,x\in\mathcal{X}$:
\begin{equation*}
    e^{-\epsilon } \le \frac{P_Y(y)}{P_{Y|X}(y|x)} \le e^{\epsilon}.
\end{equation*}

\textbf{Case 1:} When $x\neq y$:
the privacy metric $\frac{P_Y(y)}{P_{Y|X}(y|x)}$ becomes:
\begin{equation*}
    \frac{P_X(y)}{P_{Y|X}(y|x)} = \frac{P_X(y)}{P_X(y)e^{-\epsilon}} \overset{\Delta}{=} e^{\epsilon},
\end{equation*}
which is guaranteed to be bounded by $[e^{-\epsilon}, e^{\epsilon}]$.

\textbf{Case 2:} When $x = y$:
the privacy metric $\frac{P_Y(y)}{P_{Y|X}(y|x)}$ becomes:
\begin{equation*}
    \frac{P_X(y)}{P_{Y|X}(y|x)} = \frac{P_X(y)}{1 - \frac{(1-P_X(y))}{e^{\epsilon}}} = \frac{P_X(y)e^{\epsilon}}{e^{\epsilon} - 1+P_X(y)}.
\end{equation*}
Next, we compare with the upper and lower bound respectively.

Comparing with the upper bound:
\begin{equation*}\label{upper}
    \begin{aligned}
    e^{\epsilon} -\frac{P_X(y)e^{\epsilon}}{e^{\epsilon} - 1+P_X(y)}
    =&\frac{[e^{\epsilon} - 1+P_X(y) - P_X(y)]e^{\epsilon}}{e^{\epsilon} - 1+P_X(y)}\\
    =& \frac{(e^{\epsilon} - 1)e^{\epsilon}}{e^{\epsilon} - 1+P_X(y)} \ge 0.
    \end{aligned}
\end{equation*}
Comparing with the lower bound:
\begin{equation*}
\begin{aligned}
    &\frac{P_X(y)e^{\epsilon}}{e^{\epsilon} - 1+P_X(y)} - e^{-\epsilon}\\
    =&\frac{P_X(y)e^{\epsilon}  - e^{-\epsilon}(e^{\epsilon} - 1+P_X(y))}{e^{\epsilon} - 1+P_X(y)}\\
    =&\frac{P_X(y)e^{\epsilon}  - 1 + e^{-\epsilon}-P_X(y)e^{-\epsilon}}{e^{\epsilon} - 1+P_X(y)}\\
    =&\frac{P_X(y)e^{\epsilon}(1- e^{-2\epsilon})  - (1 - e^{-\epsilon})}{e^{\epsilon} - 1+P_X(y)}\\
    =&\frac{P_X(y)e^{\epsilon}(1- e^{-\epsilon})(1+e^{-\epsilon})  - (1 - e^{-\epsilon})}{e^{\epsilon} - 1+P_X(y)}\\
    =&\frac{(1- e^{-\epsilon})[P_X(y)e^{\epsilon}(1+e^{-\epsilon})  - 1]}{e^{\epsilon} - 1+P_X(y)}.\\
\end{aligned}
\end{equation*}
Note that $(1-e^{-\epsilon}) \ge 0$, $e^{\epsilon} - 1+P_X(y) \ge 0$, therefore, the lower bound holds when:
\begin{equation*}
    P_X(y)e^{\epsilon}(1+e^{-\epsilon})  - 1 \ge 0,~~ \text{or}~~ 1-e^{-\epsilon} = 0
\end{equation*}
which implies for all $y \in \mathcal{X}$
\begin{equation*}\label{eq:lower_bound}
    \begin{aligned}
    P_X(y)  \ge \frac{1}{e^{\epsilon} +1}, ~~ \text{or}~~ \epsilon = 0.
    \end{aligned}
\end{equation*}
 This further implies that the mechanism achieves $\epsilon$-LIP if for the $\min_{y\in\mathcal{X}} P_X(y)$, denoted as $P_{\min}$ in the following:
\begin{equation*}
P_{\min} \ge  \frac{1}{e^{\epsilon} +1}. 
\end{equation*}
Therefore, the mechanism achieves $\epsilon$-LIP, when 
\begin{equation*}
    \epsilon \ge \log \left(\frac{1}{P_{\min}} -1\right).
\end{equation*}

\section{Proof of Theorem 2 and Remark 2}
In this section, we proof Theorem 2, i.e., the context-aware RR mechanism achieves $(\epsilon,\delta)$-LIP, where
\begin{equation}
    \delta \ge \frac{(e^{\epsilon} -1)[1 - P_{\min}(e^{\epsilon} + 1)]}{e^{2\epsilon}}.
\end{equation}
and the corresponding valid region of $\epsilon$.
\begin{proof}
When $y\neq x$, the two inequalities in definition 1 become:

\begin{equation*}
\begin{aligned}
P_Y(y) \ge e^{-\epsilon} P_{Y|X}(y|x) - \delta
\Leftrightarrow &P_X(y) \ge P_X(y)e^{-2\epsilon} - \delta\\
\Leftrightarrow &\delta \ge \max\{0, P_X(y)(e^{-2\epsilon}-1)\}\\
\Leftrightarrow &\delta \ge 0.\\
\end{aligned}
\end{equation*}
and 
\begin{equation*}
\begin{aligned}
P_Y(y) \le e^{\epsilon} P_{Y|X}(y|x) + \delta
\Leftrightarrow &P_X(y) \le P_X(y) + \delta\\
\Leftrightarrow &\delta \ge 0.\\
\end{aligned}
\end{equation*}

This implies that when $y\neq x$, both conditions are satisfied for any $\delta \ge 0$.

When $y = x$, the right hand side inequality in Definition 1 can be simplified to:
\begin{equation*}
\begin{aligned}
P_Y(y) \le e^{\epsilon} P_{Y|X}(y|x) + \delta
\Leftrightarrow &P_X(y) \le e^{\epsilon} - 1 +P_X(y) + \delta\\
\Leftrightarrow &\delta \ge \max\{0, 1 - e^{\epsilon}\}\\
\Leftrightarrow &\delta \ge 0.
\end{aligned}
\end{equation*}
This implies that when $y= x$, both right hand side condition is satisfied for any $\delta \ge 0$.
On the other hand, for the left hand side
\begin{equation*}
P_Y(y)\ge P_{Y|X} (y|x) e^{-\epsilon} -\delta,
\end{equation*}
we have 
\begin{equation*}
    \begin{aligned}
    \delta \ge & P_{Y|X} (y|x) e^{-\epsilon} - P_Y(y)\\
    = &\frac{e^{\epsilon} - 1 + P_X(y)}{e^{2\epsilon}} - P_X(y)\\
    = & \frac{(e^{\epsilon} -1)[1 - P_{X}(y)(e^{\epsilon} + 1)]}{e^{2\epsilon}},
    \end{aligned}
\end{equation*}
which implies that 
\begin{equation*}
    \delta \ge \frac{(e^{\epsilon} -1)[1 - P_{X}(y)(e^{\epsilon} + 1)]}{e^{2\epsilon}} \ge \frac{(e^{\epsilon} -1)[1 - P_{\min}(e^{\epsilon} + 1)]}{e^{2\epsilon}}.
\end{equation*}
This further implies that 
\begin{equation}\label{eq:ep_con}
    (\delta + P_{\min})e^{2\epsilon} - e^{\epsilon} + 1 - P_{\min} \ge 0.
\end{equation}
As $\delta + P_{\min} >0$, when 
$1-4(\delta + P_{\min})(1-P_{\min})\le 0 $, i.e., $\delta > \frac{1}{4(1-P_{\min})}-P_{\min}$, condition in \eqref{eq:ep_con} holds for every $\epsilon\ge0$.
When $1-4(\delta + P_{\min})(1-P_{\min})> 0 $, i.e., $\delta \le \frac{1}{4(1-P_{\min})}-P_{\min}$,
condition in \eqref{eq:ep_con} holds when 
\begin{equation*}
   \epsilon\in \left[0,\log\left(\frac{1 - \sqrt{1-4(\delta+P_{\min})(1-P_{\min})}}{2(\delta +P_{\min})}\right)\right],
\end{equation*}
or 
\begin{equation*}
    \epsilon\in \left.\left[\log\left(\frac{1 + \sqrt{1-4(\delta+P_{\min})(1-P_{\min})}}{2(\delta +P_{\min})}\right), \infty\right.\right).
\end{equation*}


We next show that $\log\left(\frac{1 - \sqrt{1-4(\delta+P_{\min})(1-P_{\min})}}{2(\delta +P_{\min})}\right)\ge 0$ for any $P_{\min}\in[0,1]$ and $\delta\in[0,1]$, which is equivalent to show that:
\begin{equation}\label{eq:condition}
    \begin{aligned}
    &\frac{1 - \sqrt{1-4(\delta+P_{\min})(1-P_{\min})}}{2(\delta +P_{\min})}\ge 1.\\
    \Leftrightarrow& 1- 2(\delta +P_{\min}) \ge \sqrt{1-4(\delta+P_{\min})(1-P_{\min})}.\\
    \end{aligned}
\end{equation}

Note that as $\delta \le \frac{1}{4(1-P_{\min})}-P_{\min}$, 
\begin{equation*}
    \begin{aligned}
    1-2(\delta + P_{\min}) \le 1-\frac{1}{2(1-P_{\min})},
    \end{aligned}
\end{equation*}
and \begin{equation*}
    1-\frac{1}{2(1-P_{\min})}\ge 0,
\end{equation*}
when $P_{\min}\le 1/2$. This is true for any $X$ with support containing at least two instances.
When $1-2(\delta + P_{\min})\ge 0$, to satisfy the condition in \eqref{eq:condition}, it is equivalent to show that:
\begin{equation}
\begin{aligned}
    &(1- 2(\delta +P_{\min}))^2 \ge 1-4(\delta+P_{\min})(1-P_{\min})\\
    \Leftrightarrow&P_{\min}(P_{\min} + \delta) \ge 0,\\
\end{aligned}
\end{equation}
which holds for any $P_{\min}\in[0,1]$ and $\delta\in[0,1]$, 
This completes the proof of Theorem 2 and remark 2.
\end{proof}

\end{document}